\documentclass{elektr}
\usepackage{hyperref}
\hypersetup{
colorlinks=true,
urlcolor=blue,
citecolor=blue}
\usepackage[all]{xy,xypic}
\usepackage{amsfonts,amssymb,amsmath,amsgen,amsopn,amsbsy,theorem,graphicx,epsfig}
\usepackage{eufrak,amscd,bezier,latexsym,mathrsfs,enumerate}\usepackage[utf8]{inputenc}\usepackage[english]{babel}
\usepackage[dvipsnames]{xcolor}
\usepackage[pagewise]{lineno}

\year{}
\vol{}
\fpage{}
\lpage{}
\doi{}

\newcommand{\bfc}{\mbox{\boldmath $c$}}

\newcommand{\bfx}{\mbox{\boldmath $x$}}
\newcommand{\bfy}{\mbox{\boldmath $y$}}
\newcommand{\bfp}{\mbox{\boldmath $p$}}

\title{Improving the Redundancy of Knuth's Balancing Scheme for Packet Transmission Systems}

\author[Elie Ngomseu Mambou,
Ebenezer Esenogho,~and~Hendrik C. Ferreira]{
\textbf{Elie Ngomseu Mambou\thanks{emambou@uj.ac.za}~, Ebenezer Esenogho, Hendrik C. Ferreira}\\
Department of Electrical and Electronic Engineering\\
University of Johannesburg\\ P. O. BOX 524, Auckland Park, 2006, South Africa
\\ [1.8em]

\rec{.201}
\acc{.201}
\finv{..201}
}

\def\E{\ifmmode{\mathbb E}\else{$\mathbb E$}\fi} 
\def\N{\ifmmode{\mathbb N}\else{$\mathbb N$}\fi} 
\def\R{\ifmmode{\mathbb R}\else{$\mathbb R$}\fi} 
\def\Q{\ifmmode{\mathbb Q}\else{$\mathbb Q$}\fi} 
\def\C{\ifmmode{\mathbb C}\else{$\mathbb C$}\fi} 
\def\H{\ifmmode{\mathbb H}\else{$\mathbb H$}\fi} 
\def\Z{\ifmmode{\mathbb Z}\else{$\mathbb Z$}\fi} 
\def\P{\ifmmode{\mathbb P}\else{$\mathbb P$}\fi} 
\def\T{\ifmmode{\mathbb T}\else{$\mathbb T$}\fi} 
\def\SS{\ifmmode{\mathbb S}\else{$\mathbb S$}\fi} 
\def\DD{\ifmmode{\mathbb D}\else{$\mathbb D$}\fi} 

\newcommand{\bse}{\begin{subequations}}
\newcommand{\ese}{\end{subequations}}
\newcommand{\ben}{\begin{enumerate}}
\newcommand{\een}{\end{enumerate}}
\newcommand{\bens}{\begin{enumerate*}}
\newcommand{\eens}{\end{enumerate*}}
\newcommand{\be}{\begin{equation}}
\newcommand{\ee}{\end{equation}}
\newcommand{\bea}{\begin{eqnarray}}
\newcommand{\eea}{\end{eqnarray}}
\newcommand{\baa}{\begin{eqnarray*}}
\newcommand{\eaa}{\end{eqnarray*}}
\newcommand{\bc}{\begin{center}}
\newcommand{\ec}{\end{center}}

\newtheorem{theorem}{Theorem}

\theoremstyle{corollary}

\theoremstyle{lemma}
\newtheorem{lemma}{Lemma}

\theoremstyle{proposition}

\theoremstyle{axiom}

\theoremstyle{conjecture}

\theoremstyle{example}
\newtheorem{example}{Example}

\theoremstyle{definition}

\theoremstyle{remark}

\begin{document}

\maketitle

\begin{abstract}A simple scheme was proposed by Knuth to generate binary balanced codewords from any information word. However, this method is limited in the sense that its redundancy is twice as that of the full sets of balanced codes. The gap between Knuth’s algorithm redundancy and that of the full sets of balanced codes is significantly considerable. This paper attempts to reduce that gap. Furthermore, many constructions assume that a full balancing can be performed without showing the steps. A full balancing refers to the overall balancing of the encoded information together with the prefix. We propose an efficient way to perform full balancing scheme which do not make use of lookup tables or enumerative coding. 

\end{abstract}

\section{Introduction}
Balanced codes have been widely studied over the years because of their applicability in the field of communication and storage structures such as optical and magnetic recording devices like Blu-Ray, DVD and CD \cite{immink1986,leiss1984}; error correction and detection \cite{bassam1993, weber2011}; cable transmission \cite{cattermole1983} and noise attenuation in VLSI integrated circuits \cite{tabor1990}. For some balancing techniques, the decoding of balanced codes is fast and can be done in parallel which avoids latency in communication.

A binary word of length k, with k even, is said to be balanced if the number of zeros and ones equals $k/2$. Knuth proposed a simple and efficient scheme to generate balanced codewords \cite{knuth1986}. This approach stipulates that any binary word $\bfx$ of length $k$ can always be encoded into a balanced one denoted as $\bfy$ by inverting the first $e$ bits of $\bfx$ where $1 \leq e \leq k$. The index $e$ is encoded as the prefix $\bfp$ which is prepended to y and send through a channel. At the receiver, the decoder receives the codeword $\bfp\bfy$ and retrieves the original information word through the prefix by inverting back the first $e$ bits of $\bfy$. This algorithm is very suitable for long sequences as it does not make use of any lookup tables neither at the encoder nor at the decoder. A detailed explanation of this method is covered in Section \ref{sec2}.

A lot of works have been done to reduce the redundancy generated by Knuth’s algorithm (KA). In \cite{weber2010}, two attempts were described by Weber and Immink; the first one was using the distribution of the prefix index. This consists of setting the encoder to choose smaller values for the prefix index knowing that the position index e might not be unique. By default, KA makes use of the first balancing index while inverting from least position bits. It has been shown that the distribution of that index for equiprobable information words is not uniform and presents a redundancy which is slightly less than that of KA. The second attempt used the multiplicity of balancing points within a word to transmit auxiliary data. The previous schemes provide a fixed length (FL) and variable length (VL) prefix implementation. However, these methods only made a minor improvement on KA.

The second attempt from \cite{weber2010} was exploited in \cite{al2013} and renamed as bit recycling for Knuth’s algorithm (BRKA); it relies on a high probability of having more than one balancing index while performing KA; in other words, this scheme uses the multiplicity of balancing indexes to encode a shorter prefix than that from KA. In \cite{dube2017}, a technique for balancing words was presented based on permutations, the arcade game Pacman and limited-precision integers; the redundancy of KA was improved and the redundancy of the full set of balanced codes was nearly achieved at the cost of high complexity and large memory usage.

A major contribution was brought by Immink and Weber \cite{immink2010} through an efficient encoding of the index prefix for both VL and FL schemes. This scheme is based on associating distinctly each word of a code to a balanced codeword. More details on this method will be provided. Furthermore, the distribution of the prefix length was discussed as well as the algorithm complexity.

In this paper, a modification of a scheme described in \cite{immink2010} is proposed to generate efficient and less-redundant balanced codes compared to most state-of-the-art techniques. This approach is designed for communication systems that model the data as packets contrarily to cascade-based model. The rest of this paper is structured as follows: A background study is done in Section \ref{sec2}; the system model of the proposed scheme is described in Section \ref{sec3}; then in Section \ref{sec4}, the proposed encoding is presented. Section \ref{sec5} and \ref{sec6} provide detailed analysis as well as performance and discussions of the proposed scheme redundancy. Finally, the paper is concluded in Section \ref{sec7}.

\section{Background}\label{sec2}
Let $\bfx=(x_1x_2\dots x_k)$ be a bipolar sequence of length $k$ and $\bfp=(p_1p_2\dots p_{r})$, the prefix of length $r$.  $\bfc=(c_1c_2\dots c_{n})$ of length $n=k+r$ is the transmitted codeword comprised of the encoding of $\bfx$ denoted as $\bfy$ and appended to $\bfp$, $\bfc=(\bfp\cdot\bfy)$. All these words are defined within the alphabet $\mathcal{A}^2$ where $\mathcal{A}^2= \{-1,1\}$. Let $d(\bfx)$ refer to the sum of all digits in $\bfx$, also called the disparity of $\bfx$. The word $\bfx$ is said to be balanced if 
$d(\bfx)=\sum_{i=1}^{k}x_i=0.$

Similarly, the disparity of the first $j$ bits of $\bfx$, also called running digital sum (RDS), is denoted as $d_j(\bfx)$; and $d_j(\bfx)=\sum_{i=1}^{j}x_i$, where $1\leq j\leq k$. For the scope of this paper, the information word length is considered to be even.

\subsection{Knuth's balancing scheme}\label{2a}
The celebrated Knuth's scheme consists of complementing a word bits up to certain point. This is equivalent to splitting a word into two segments, the first one has its bits flipped and the second is unchanged. It was shown in \cite{knuth1986} that this simple and efficient procedure will always generate at least one balanced codeword. If $e$ is the index of the first balancing point then, the disparity of $\bfx$ is given by:
\begin{equation}
d(\bfx)=-\sum_{i=1}^{e}x_i + \sum_{i=e+1}^{k}x_i.
\end{equation} 

Those summations reflect the two segments that build a balanced codeword.  Because $d_{j+1}(\bfx)=d_j(\bfx)\pm 2$, it is always achievable to find an index $e$ corresponding to a balancing point such that $d(\bfx)=0$. The index $e$ might be unique, by convention, the Knuth's algorithm only considers the first one while inverting from least index bits. In parallel scheme, the index $e$ is encoded as the prefix and prepended to $\bfy$. The length of the prefix, $r$, is given by:
\begin{equation}
r=\lceil\log_2k\rceil, \mbox{ for }k\gg 1.
\end{equation}

The redundancy of a full set of balanced codewords of length $k$, denoted as $H_0(k)$ equals
\begin{equation}\label{h0}
H_0(k)=k-\log_2{k\choose k/2},
\end{equation}

An approximation of $H_0(k)$ was given in \cite{knuth1986} as
\begin{equation}\label{eqh0}
H_0(k)\approx \frac{1}{2}\log_2k+0.326, \mbox{ for }k\gg 1.
\end{equation}

For large $k$, the Knuth's scheme redundancy is almost twice as large as $H_0(k)$. 
\subsection{Efficient binary balanced codewords}\label{2b}

Let $\bfx^j$ be the word $\bfx$ where first $j$ bits are inverted. If $e$ represents the index of the first balancing point then $\bfy=\bfx^e$ is the balanced codeword through Knuth's scheme. There are $k$ different ways of inverting the word $\bfx$. In \cite{immink2010}, it was established that some words from the set ${\bfx^1, \bfx^2,\dots,\bfx^k}$ can be associated with the balanced word $\bfy$ following the Knuth's scheme. Let $s(\bfy)$ be the set of all words associated with a balanced codeword, $\bfy$, $s(\bfy)=\{\bfx: \bfx^j=\bfy \mbox{ with }1\leq j\leq k\}$; and $|s(\bfy)|$, its cardinality.

The prefix of the encoded word corresponds to the information word rank within the subset $s(\bfy)$.  It was shown in \cite{immink2010} that the size of $s(\bfy)$ is such that: $2\leq |s(\bfy)| \leq \frac{k}{2}+1$, where $|s(\bfy)|=\mbox{max}\{d_j(\bfx)\}-\mbox{min}\{d_j(\bfx)\}+1$ with $\mbox{max}\{d_j(\bfx)\}$ and $\mbox{min}\{d_j(\bfx)\}$ being the maximum and minimum RDS values of $\bfx$ respectively. For the fixed length (FL) scheme, the prefix has exactly $\log_2\big(\frac{k}{2}+1\big)$ bits while in the variable length (VL) scheme, the prefix length varies between 1 and $\log_2\big(\frac{k}{2}+1\big)$ bits.

\section{System Model}\label{sec3}
Fig. \ref{fig:sys} shows a model of communication for two different systems. In Fig.~\ref{fig:sys}(a), the data is received as a set of balanced codewords; in this model, the decoder must keep track of start and ending of each data block for the purpose of synchronization that relies on prefixes. 
\begin{figure}[h!]
	\centering
	\includegraphics[width=.9\linewidth]{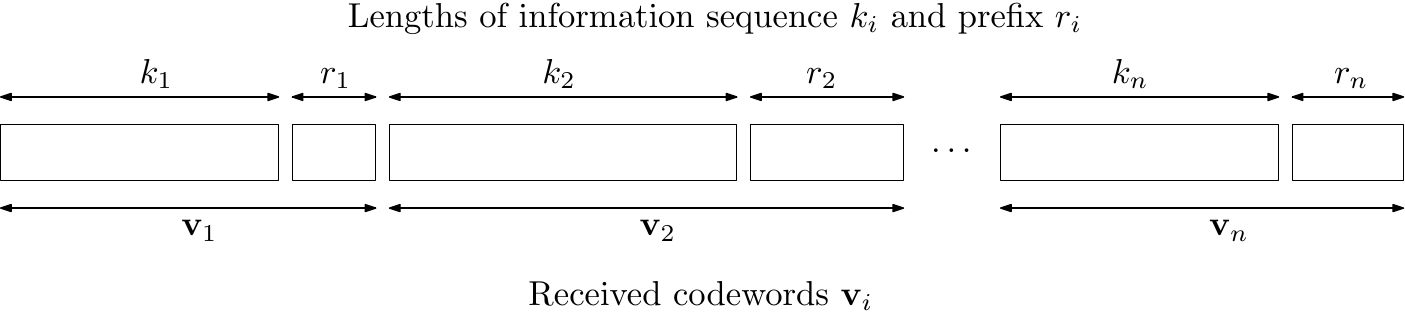}	\\
	(a)		\\
	\includegraphics[width=0.35\linewidth]{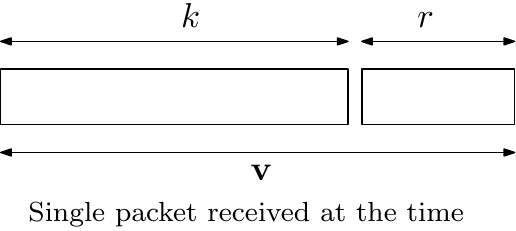}	\\
	(b)		
	\caption{System model (a) Cascade-based   vs. (b) Packet-based}
	\label{fig:sys}
\end{figure}

Whereas, in Fig.~\ref{fig:sys}(b), the packet conception represents a single data block which is received once at the time. This concept is used in  various communications systems such as Bluetooth/wireless communication, smart grids systems, GSM networks, power line communication (PLC), visible light communication (VLC), network communication, etc. The communication is incremental as packets are transmitted once at the time and an ``ACK" message is required before a subsequent packet is sent; in this case, the decoder only keeps track on the end of packet (EOP). In other words, the packet-based communication model can be viewed as a ternary alphabet, $\{0,1,\epsilon\}$, where $\epsilon$ represents the EOP marker whereas the cascade-based model uses binary communication. We exploit the advantage of the ternary alphabet to derive an efficient and less redundant scheme for generating balanced codewords. For the rest of this work, the packet-based model is considered.

\section{Encoding scheme}\label{sec4}
This consists of associating every information sequence of length $k$ with a balanced codeword using the Knuth's inversion rule as described in \cite{immink2010}. This leads to ${k\choose \frac{k}{2}}$ distinct sets. Furthermore, each of these sets is compressed to decrease the redundancy based on Lemma \ref{th2} .  

Let $\bfx$ be the information sequence to be encoded, if $\bfx$ is already balanced, a protocol is adopted between the transmitter and the receiver to have a prefix-less codeword; otherwise, $\bfx$ is balanced following the Knuth's algorithm, associated with the corresponding balanced codeword $\bfy$. Following the same procedure, all information sequences can be associated with balanced codewords, $\bfy$ and listed according to the lexicographic order within subsets $s(\bfy)$. The prefix of each $\bfx$ corresponds to its rank within $s(\bfy)$.
\begin{lemma}\label{th2}
	Any balanced codeword is always associated with another balanced one. 
\end{lemma}

\begin{proof}
	This is an observation from the Knuth's algorithm structure; any already balanced codeword always results in another balanced one; at the worst case scenario,  the balanced state is always found by inverting all bits of an already balanced codeword.\end{proof}

Let $|s(\bfy)|$ denote the cardinality of the subset $s(\bfy)$ that comprises all associated sequences with $\bfy$. The inclusion of balanced sequences within the set of information sequence as presented in \cite{immink2010}, adds an extra rank in the rank position in every set; this is useful for cascade-based model. However, as described in Lemma \ref{th2}, a balanced sequence is always associated to another one. This important observation leads to the compression of $s(\bfy)$ by removing the already balanced sequence. This is suitable for the packet-based model as described in Fig. \ref{fig:sys}(a).

\begin{example}\label{ex2}
	Let us consider binary sequences of length $k=4$. 
	\begin{equation}\small\label{t1}
	\begin{tabular}{crcccccc}
	$\bfy$ & 0011 & 0101  & 0110 & 1001 & 1010 & 1100&\bfp \\
	\hline 
	$s(\bfy)$ & \textcircled{\tiny{1}}1011& 1101 & 1000 & 0001 & 0010 & 0000&00 \\ 
	& \textcircled{\tiny{2}}1111 &\textbf{1001}& 1110 & 0111 &\textbf{ 0110} & 0100&01 \\
	& \textcircled{\tiny{3}}\textbf{1100}&  &\textbf{1010}&\textbf{0101}&&\textbf{0011}&10
	\end{tabular} 
	\end{equation}
	$\textcircled{\tiny{1}} 1011\rightarrow 0011$\\
	$\textcircled{\tiny{2}} 1111\rightarrow 0111\rightarrow 0011$\\
	$\textcircled{\tiny{3}} 1100\rightarrow 0100\rightarrow 0000\rightarrow 0010\rightarrow 0011$
	
	Eq.~\eqref{t1} shows the encoding process described in \cite{immink2010}, whereby balanced codewords (marked in bold) are part of subsets $s(\bfy)$. Lines $\textcircled{\tiny{1}}$, $\textcircled{\tiny{2}}$ and $\textcircled{\tiny{3}}$ show how balanced codewords are obtained from the Knuth's balancing scheme. $\bfp$ represent prefixes.
	
	\begin{equation}\small\label{t2}
	\begin{tabular}{cccccccc}
	$\bfy$ & 0011 & 0101  & 0110 & 1001 & 1010 & 1100&\bfp \\
	\hline 
	$s(\bfy)$ & \textcircled{\tiny{1}}1011& 1101 & 1000 & 0001 & 0010 & 0000&0 \\ 
	& \textcircled{\tiny{2}}1111 && 1110 & 0111 & & 0100&1 \\
	\end{tabular} 
	\end{equation}
	
	Eq. \eqref{t2} presents the proposed encoding process where all subsets do not include balanced sequences.
\end{example}

The cardinality of the subset $s(\bfy)$ can be derived from RDS calculations on the balanced codeword, $\bfy$ as presented in Lemma \ref{th3}.
\begin{lemma}\label{th3}
	$|s(\bfy)|=\text{max}\{d_j(\bfy)\} - \text{min}\{d_j(\bfy)\}.$
\end{lemma}

\begin{proof}
	It was proved in \cite{immink2010} that $|s(\bfy)|=\mbox{max}\{d_j(\bfy)\}-\mbox{min}\{d_k(\bfy)\}+1$; the balanced codeword was removed out of every set. So, the new $|s(\bfy)|$ is subtracted by 1, leading to $|s(\bfy)|=\mbox{max}\{d_j(\bfy)\}-\mbox{min}\{d_j(\bfy)\}.$   	
\end{proof}

For any subset $s(\bfy)$, its size is always bounded as per Theorem \ref{th4}.
\begin{theorem}\label{th4}
	$1\leq |s(\bfy)|\leq \frac{k}{2}.$
\end{theorem}

\begin{proof}
	It was established in \cite{immink2010} that $2\leq |s(\bfy)|\leq \frac{k}{2}+1$; then after removing the balanced codeword out of every set, it follows that $1\leq |s(\bfy)|\leq \frac{k}{2}$.
\end{proof}

Therefore, the required fixed prefix length for this scheme is $\log_2\frac{k}{2}$; this is a slight improvement on the Knuth's scheme that has a redundancy of $\log_2k$ as well as on the scheme in \cite{immink2010} where it equals $\log_2(\frac{k}{2}+1)$. In addition, prefixes are obtained from ranking the information sequences associated to a balanced codeword from $0$ to $\frac{k}{2}-1$.

\section{Study of Sparseness of $|s(\bfy)|$}\label{sec5}
Let $N(\lambda, k)$ be the number of possible balanced codewords $\bfy$ of length $k$ such that $|s(\bfy)|=\lambda$. 

The following equation holds from Theorem \ref{th4}:
$$\sum_{\lambda=1}^{k/2} N(\lambda, k) = {k\choose \frac{k}{2}}.$$

For the convenience of the reader, details on computing $N(\lambda, k)$ for $1\leq \lambda \leq \frac{k}{2}$ are derived by following the guidelines as in \cite{immink2010}. 

The derivation of $N(\lambda,k)$ was done using the computation of the number of bipolar sequences whose running sum lies within two finite bounds $B1$ and $B2$ (with $B2>B1$) \cite{chien1970}. 

The interval of running sum values that a sequence may reach, also referred to as the \textit{digital sum variation (DSV)} is given by $B=B2-B1+1$. Each iteration in the random walk of a sequence defines an entry of an $B\times B$ connection matrix, $M_B$.

$M_B$ is such that, $M_B(i,j)=1$, if there is a path in the random walk from state $s_i$ to state $s_j$; and $M_B(i,j)=0$ if no path can be established. For each iteration, a random walk of the running sum can only move one state up or down.  Therefore, $M_B(i+1,i)=M_B(i,i+1)=1$ and $M_B(i,j)=0$, where $i,j=1,2,\dots,B-1$ as presented in (\ref{eqm}). 

\begin{equation}
M_B=\begin{bmatrix}
0 & 1& 0 & \dots  &0&0 \\
1 & 0 & 1 & \dots  &0& 0 \\
0 & 1 & 0 & \ddots  &0& 0 \\
\vdots & \vdots & \ddots & \ddots & \vdots &\vdots\\
0 & 0 & 0 & \dots  &0& 1 \\
0& 0 & 0 & \dots  &1& 0
\end{bmatrix}
\label{eqm}
\end{equation}
$M_B^k(i,j)$ denotes the $(i,j)^{\mbox{th}}$ entry of the $k^{\mbox{th}}$ power of $M_B$.

\begin{theorem}
	The number of balanced codewords $\bfy$ of length $k$ and $|s(\bfy)|=\lambda$, $N(\lambda, k)$ for $1\leq \lambda\leq \frac{k}{2}$, is such that
	$$ N(\lambda, k)= \sum_{i=1}^{\lambda+1}M_{\lambda+1}^k(i,i)-2\sum_{i=1}^{\lambda}M_\lambda^k(i,i)+\sum_{i=1}^{\lambda-1}M_{\lambda-1}^k(i,i).$$
\end{theorem}

\begin{proof}
	The number of balanced codewords such that $|s(\bfy)|=\lambda'$ for $2\leq\lambda'\leq\frac{k}{2}+1$ in \cite{immink2010} was  $$ N(\lambda', k)= \sum_{i=1}^{\lambda'}M_{\lambda'}^k(i,i)-2\sum_{i=1}^{\lambda'-1}M_{\lambda'-1}^k(i,i)+\sum_{i=1}^{\lambda'-2}M_{\lambda'-2}^k(i,i).$$	
	
	Therefore, for $1\leq \lambda\leq \frac{k}{2}$, the set of all random walks between bounds $B2$ and $B1$ is shifted one unit down. This leads to $N(\lambda,k)$ balanced codewords where $\lambda=\lambda'-1$. 
	
	This leads to the following
	$$ N(\lambda, k)= \sum_{i=1}^{\lambda+1}M_{\lambda+1}^k(i,i)-2\sum_{i=1}^{\lambda}M_\lambda^k(i,i)+\sum_{i=1}^{\lambda-1}M_{\lambda-1}^k(i,i).$$
\end{proof}
A simplified expression of $M_B$ was provided in \cite{immink2010} based on a formula to compute powers of $M_B$ derived by Salkuyeh \cite{salkuyeh2006} as follows:
\begin{equation}
\sum_{i=1}^{B}M_B^k(i,i)=2^k\sum_{i=1}^{B}\cos^k\frac{\pi i}{B+1}.
\end{equation}
This makes the computation of $N(\lambda,k)$ much simpler as follows:
\begin{equation}
N(\lambda,k)=2^k\Big( \sum_{i=1}^{\lambda+1}\cos^k\frac{\pi i}{\lambda+2}-2\sum_{i=1}^{\lambda}\cos^k\frac{\pi i}{\lambda+1}+\sum_{i=1}^{\lambda-1}\cos^k\frac{\pi i}{\lambda}\Big).
\label{eqcomp}
\end{equation}

The computation of $N(\lambda,k)$ as presented in (\ref{eqcomp}) becomes obvious for special values of $\lambda$ as shown in (\ref{a2}). The enumeration of sequences corresponding to these values of $\lambda$ as well as the pseudo code for computing $|s(\bfy)|$, for generating the ordered set of information sequences and determining the prefix index were provided in \cite{immink2010}.
\begin{equation}\label{a2}
\begin{tabular}{c|l}
$\lambda$& $N(\lambda, k)$ \\ 
\hline 
1& 2\\
2& $2(2^{\frac{k}{2}-1})$\\
$\frac{k}{2}-1$& $k(k-4)$, $k>4$\\
$\frac{k}{2}$&$k$
\end{tabular}
\end{equation}

\section{Analysis and Discussions}\label{sec6}
In this section, the average number of bits denoted as $H(k)$ required to encode the prefix index of a sequence of length $k$ is computed. The number of all information sequences associated with balanced codewords is $2^k-{k\choose\frac{k}{2}}$.
\begin{equation}\label{a3}
\sum_{\lambda=1}^{k/2}\lambda N(\lambda,k) = 2^k-{k\choose \frac{k}{2}}.
\end{equation}
It follows that
\begin{equation}\label{eqh}
H(k)=\frac{\sum_{\lambda=1}^{k/2}\lambda N(\lambda,k) \log_2\lambda}{2^{k}-{k\choose \frac{k}{2}}}
\end{equation}
The average number of bits for the construction in \cite{immink2010} is as follows:
\begin{equation}\label{eqh1}
H_1(k)=2^{-k}\sum_{\lambda=2}^{\frac{k}{2}+1}\lambda N(\lambda,k)\log_2\lambda.
\end{equation}
The average number of bits for the method in \cite{al2013} is given by
\begin{equation}\label{eqh2}
H_2(k)=\sum_{c=1}^{\frac{k}{2}}P(c) AV(c),
\end{equation}

\noindent where \\
$P(c)=2^{c+1-k} {k-1-c\choose \frac{k}{2}-c}$, $1\leq c\leq \frac{k}{2}$, $d=c-2^{\left \lfloor{\log_2c}\right \rfloor}$, 
\\and\\
$AV(c)=(c-2d).\left \lfloor{\log_2c}\right \rfloor.$$\frac{1}{2^{\left \lfloor{\log_2c}\right \rfloor}}+$ $2d.\frac{1}{2^{\left \lceil{\log_2c}\right \rceil}}.\left \lceil{\log_2c}\right \rceil$.\\

Table \ref{tb:comp} presents the comparison of the average number of bits necessary to encode the prefix from various schemes. Let $d_{H_a, H_b}$ be the difference between the average prefix length $H_a$ and $H_b$; we observed that $d_{H, H_0}\leq0.61$, $d_{H, H_1}\leq0.64$ and $d_{H_2, H}\leq1.23$. 
\begin{table}[h!]
	\centering
	\caption{comparison of the prefix' average number of bits }
	\begin{tabular}{c|c|c|c|c}
		$k$& $H_0$ & $H$ & $H_1$ & $H_2$ \\
		\hline 
		4&  1.4150&  0.8000&1.4387 & 0.5000  \\ 
		8&  1.8707&  1.4632&  1.8985& 0.9375  \\ 
		16& 2.3483&   2.0806&  2.3790&  1.3706\\ 
		32&  2.8370&  2.6629&  2.8691& 1.8082 \\ 
		64&  3.3314&  3.2207&  3.3641& 2.2516 \\ 
		128&  3.8286&  3.7615&  3.8616& 2.7039 \\ 
		256&  4.3272&  4.2902&  4.3603&  3.1647\\ 
		512&  4.8265&  4.8104&  4.8597&  3.6330\\ 
		1024&  5.3261& 5.3246&  5.3594&  4.1082\\ 
	\end{tabular} 
	\label{tb:comp}
\end{table}

Fig. \ref{fig:H's} shows the comparison between the average redundancy for balanced prefixes for $H(k)$ and $H_1(k)$, denoted as $H'(k)$ and $H_1'(k)$ respectively as well as $\log_2(k)$ and $\lceil{\log_2(k)} \rceil$. $H'(k)$ is obtained from a simple modification of $H(k)$ provided in (\ref{eqh}) as follows
\begin{equation}\label{eqh'}
H'(k)=\frac{\sum_{\lambda=1}^{k/2}\lambda N(\lambda,k) \Delta(\lambda)}{2^{k}-{k\choose \frac{k}{2}}}
\end{equation}
Similarly, $H_1'(k)$ is derived from $H_1(k)$ given in (\ref{eqh1}) as follows:
\begin{equation}\label{eqh'1}
H'_1(k)=2^{-k}\sum_{\lambda=2}^{\frac{k}{2}+1}\lambda N(\lambda,k)\Delta(\lambda).
\end{equation}
Where $\Delta(\lambda)$ corresponds to the smallest value of length $k$ such that ${k\choose \frac{k}{2}}\ge \lambda$.
\begin{figure}[h!]
	\centering
	\includegraphics[width=0.9\linewidth]{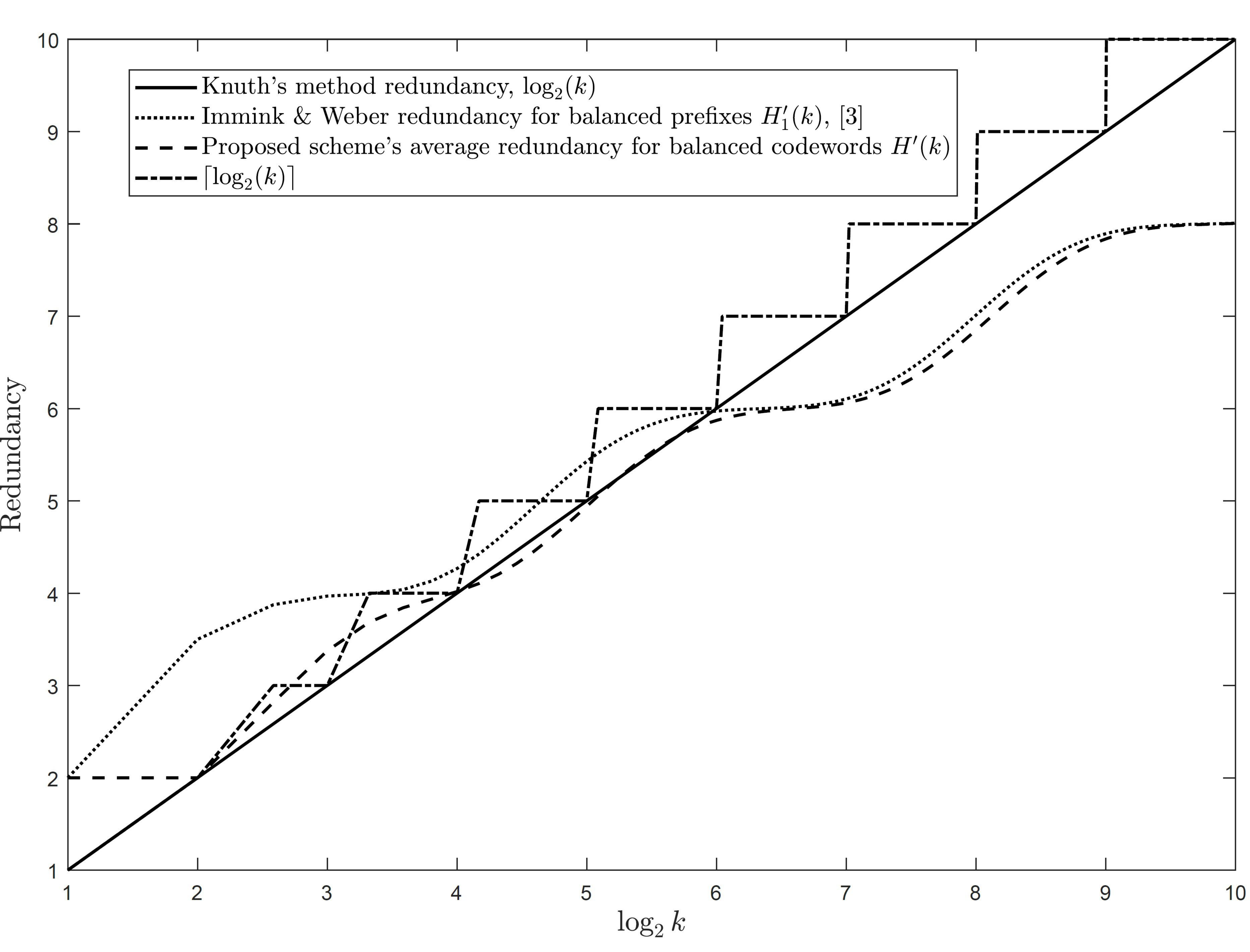}
	\caption{$H'(k)$, $H_1'(k)$, $\log_2(k)$ and $ \lceil{\log_2(k)} \rceil $ vs $\log_2k$.}
	\label{fig:H's}
\end{figure}

The graphs of $\log_2(k)$ and $\lceil{\log_2(k)} \rceil$ represents the minimum redundancy and that of integer valued redundancy of the traditional Knuth's construction. We observe that, it is only from $k>64$ that the average redundancy of the scheme presented in \cite{immink2010} is less than that of the Knuth scheme; whereas for the proposed construction, the average redundancy becomes advantageous as soon as $k>16$. Furthermore, the proposed scheme out-performs \cite{immink2010} at least for $k<1024$. 

According to Theorem \ref{th4}, the two coding schemes are applicable for the proposed scheme. For the FL prefix construction, the encoding of the prefix requires exactly $\log_2(\frac{k}{2})$ bits representing the balanced index $e$ ranging from 0 to $\frac{k}{2}-1$; whereas for the VL scheme, the prefix length varies between $0$ and $\log_2(\frac{k}{2})$ depending on the nature of the information to be encoded. A zero-prefix is used when the information sequence is already balanced. However, the VL scheme is more efficient than the FL one on the average basis.
\begin{figure}[h!]
	\centering
	\includegraphics[width=0.9\linewidth]{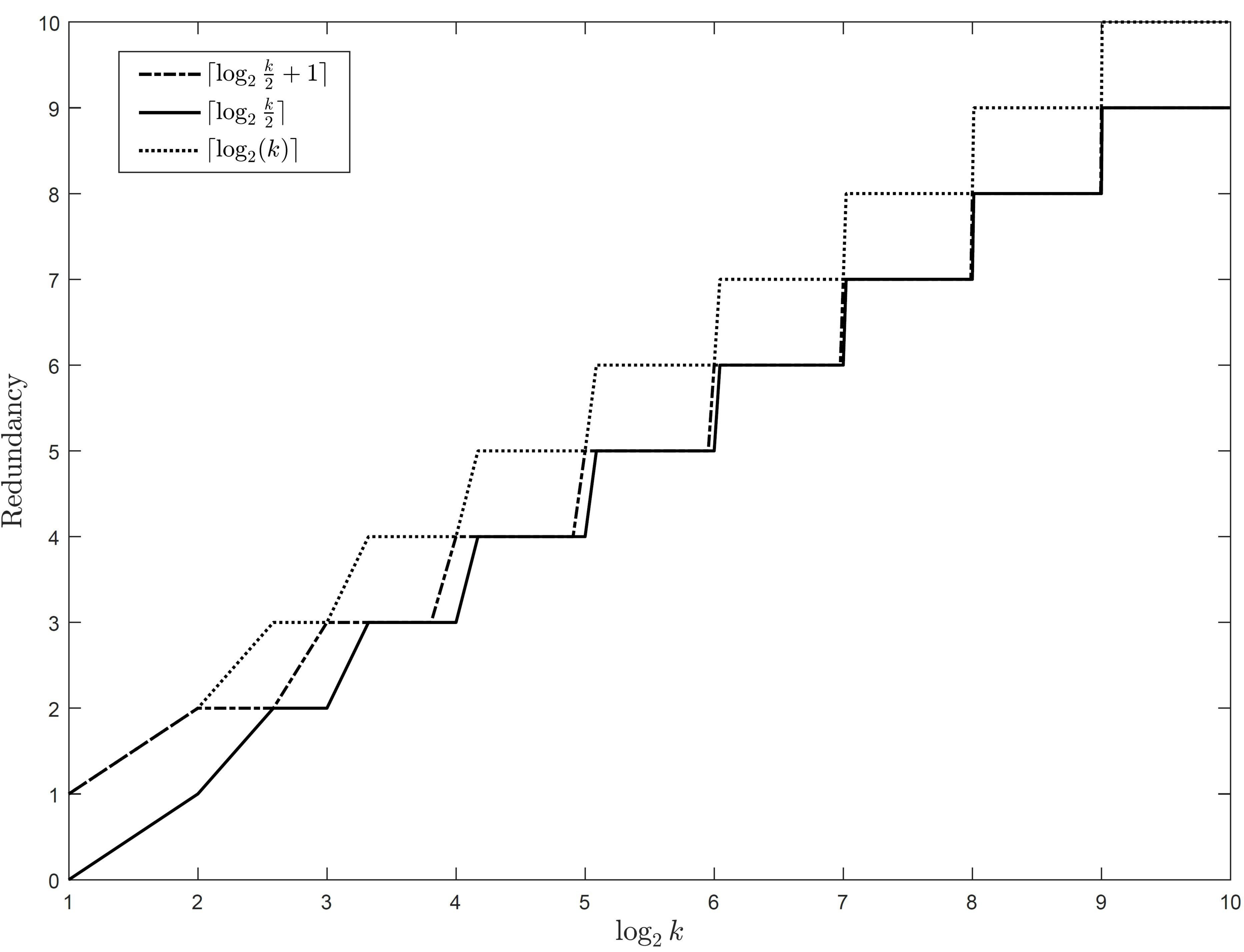}
	\caption{Rounded up FL prefix schemes}
	\label{fig:fixedCeils}
\end{figure}

In practical systems purpose, a redundancy length $r$, can only be a positive integer value. For \cite{immink2010}, $r=\lceil{\log_2(\frac{k}{2}}+1) \rceil$; in \cite{knuth1986}, $r=\lceil{\log_2(k)} \rceil$ and for the proposed construction  $r \lceil{\log_2(k)} \rceil$. Fig \ref{fig:fixedCeils} presents the rounded up FL prefix schemes. This shows that the proposed FL prefix scheme is more efficient than that of \cite{immink2010} for at least $k<512$. This improvement on short length is a great advantage for most communication systems as they make use of short packet sizes to convey information through various channels to avoid latency. 

\section{Overall Balancing}\label{sec7}
In the objective of deviating from the traditional lookup tables which memory consuming in coding, we propose a 4B6B coding based on KA. Table \ref{tab1} presents our proposed 4B6B coding which does not make use of lookup tables; red digits represent the inverted portion and bold digits are the position of balancing point index. We can notice that the table is divided into two parts: the first one consists of all inputs starting with a `0' having corresponded balanced codewords starting with a `1' and the second is all inputs starting with a `1' and having their match balanced codewords starting with a `0'.

\begin{table}[h!]
	\centering
	\caption[Proposed RLL 4B6B based on KA]{Proposed RLL 4B6B based on KA}
	\begin{tabular}{|c|c||c|c|}
		\hline 
		\textbf{Input}&\textbf{Balanced}  &  \textbf{Input}&\textbf{Balanced} \\  
		\hline 
		\textcolor{red}{00}00&  1100\textbf{10}& \textcolor{red}{0}100 & 1100\textbf{01}\\
		\hline 
		\textcolor{red}{0}001&  1001\textbf{01}&  \textcolor{red}{01}01& 1001\textbf{10}\\ 
		\hline 
		\textcolor{red}{0}010&  1010\textbf{01}&  \textcolor{red}{01}10& 1010\textbf{10} \\ 
		\hline 
		\textcolor{red}{001}1& 1101\textbf{00} & \textcolor{red}{0111} &1000\textbf{11}\\
		\hline
	\end{tabular} 
	
	\vspace*{1 cm}
	\begin{tabular}{|c|c||c|c|}
		\hline
		\textbf{Input}&\textbf{Balanced}  &  \textbf{Input}&\textbf{Balanced}  \\ 	\hline
		\textcolor{red}{1000} & 0111\textbf{00} & \textcolor{red}{110}0 & 0010\textbf{11} \\ \hline	  
		\textcolor{red}{10}01 & 0101\textbf{10} & \textcolor{red}{1}101 & 0101\textbf{01} \\ \hline	  
		\textcolor{red}{10}10 & 0110\textbf{10} & \textcolor{red}{1}110 & 0110\textbf{01} \\ \hline	  
		\textcolor{red}{1}011 & 0011\textbf{01} & \textcolor{red}{11}11 & 0011\textbf{10} \\ \hline	  
	\end{tabular}
	\label{tab1}
\end{table}

The encoder prefix consists of encoding every 4 digits into 6; prefix with length different from multiple of 4 should be filled up with `0's which leads to a FL scheme.  The input word is inverted up to a certain index which is appended at the end of the word according to the following rules which are embedded in the decoder: If the input starts with a ‘0’, the index positions, $e$ with $1 \leq e \leq 4$ are as follows: $1 \rightarrow 01$, $2 \rightarrow 10$, $3 \rightarrow 00$ and $4 \rightarrow 11$. For an input starting with a `1', the index positions are: $1 \rightarrow 01$, $2 \rightarrow 01$, $3 \rightarrow 11$ and $4 \rightarrow 00$. Therefore, an overall balancing can be achieved by encoding the prefix through the proposed 4B6B coding from Table \ref{tab1}.

\section{Conclusion}\label{sec8}

A modification of the construction given in \cite{immink2010} was proposed for packet transmission systems. The proposed scheme requires exactly $\log_2(\frac{k}{2})$ bits for the FL prefix and a prefix length between 0 and $\log_2(\frac{k}{2})$ bits for the VL scheme. The sparseness of $|s(\bfy)|$ was studied and the average efficiency of this scheme was discussed and compared to existing ones. The proposed construction is more efficient and less redundant than various schemes; it does not make use of look-up tables or enumerative coding.. Future works include a further compression of the prefix length for overall balancing through advanced efficient constructions such as \cite{mambou2016,mambou2017, mambou2017j}. On the other hand, we can extend the proposed algorithm to investigate the efficient balancing of $q$-ary sequences as higher alphabets especially powers of two's as they present various advantages in communication systems in terms of reducing latency, improving communication speed and increasing robustness and reliability \cite{ulrich1957, berrou2001, pfleschinger2001, karzand2010}.

\section*{Acknowledgement}
The authors would like to acknowledge Jos Weber for
proofreading this article and for constructive discussions. This
work has been supported partially by the Global Excellence
Stature program.

\end{document}